\makeatletter
\RequirePackage{keyval}
\define@key{Gm}{twosideshift}{%
	\ifdim#1=0pt %
	\else
	\PackageWarning{geometry}{Option `twosideshift` is gone since version 5}%
	\fi
}
\makeatother
\providecommand*{\StandardLayout}{}

\documentclass[mmnp]{my-edpsmath}

\doi{2020040}

\usepackage{amsmath,amsthm,amssymb,amsfonts}
\usepackage{graphicx}
\usepackage[labelfont=bf]{caption}
\usepackage{latexsym}
\usepackage{url}
\newtheorem{remark}{Remark}[section]

\newtheorem{theorem}{Theorem}

\begin{document}

\title{A stochastic time-delayed model for the effectiveness\\ of Moroccan COVID-19 deconfinement strategy}

\thanks{This research is part of first  author's Ph.D. project,
which is carried out at University of Aveiro.}

\thanks{Work partially supported through FCT, project UIDB/04106/2020 (CIDMA).}

\runningtitle{A stochastic time-delayed model for the effectiveness of Moroccan COVID-19}


\author{Houssine Zine}
\address{Center for Research and Development in Mathematics and Applications (CIDMA),
Department of Mathematics, University of Aveiro, 3810-193 Aveiro, Portugal;
\email{zinehoussine@ua.pt\ \&\ delfim@ua.pt}}

\author{Adnane Boukhouima}
\address{Laboratory of Analysis, Modeling and Simulation (LAMS),
Faculty of Sciences Ben M'sik, Hassan II University of Casablanca,
P.B 7955 Sidi Othman, Casablanca, Morocco;
\email{adnaneboukhouima@gmail.com\ \&\
lotfiimehdi@gmail.com\ \&\ \\
marouane.mahrouf@gmail.com\ \&\
nourayousfi.fsb@gmail.com}}

\author{El Mehdi Lotfi}
\sameaddress{2}

\author{Marouane Mahrouf}
\sameaddress{2}

\author{Delfim F. M. Torres}
\sameaddress{1}

\author{Noura Yousfi}
\sameaddress{2}

\runningauthors{H. Zine, A. Boukhouima, E. M. Lotfi, M. Mahrouf, D. F. M. Torres, N. Yousfi}

\date{Submitted: May 16, 2020; Revised: August 20, 2020; Accepted: October 28, 2020}


\begin{abstract}
Coronavirus disease 2019 (COVID-19) poses a great threat
to public health and the economy worldwide. Currently,
COVID-19 evolves in many countries to a second stage,
characterized by the need for the liberation of the economy
and relaxation of the human psychological effects. To this end,
numerous countries decided to implement adequate deconfinement strategies.
After the first prolongation of the established confinement,
Morocco moves to the deconfinement stage on May 20, 2020.
The relevant question concerns the impact on the COVID-19
propagation by considering an additional degree of realism related
to stochastic noises due to the effectiveness level of the adapted measures.
In this paper, we propose a delayed stochastic mathematical model
to predict the epidemiological trend of COVID-19 in Morocco after the deconfinement.
To ensure the well-posedness of the model, we prove the existence and uniqueness
of a positive solution. Based on the large number theorem for martingales,
we discuss the extinction of the disease under an appropriate threshold parameter.
Moreover, numerical simulations are performed in order to test the efficiency
of the deconfinement strategies chosen by the Moroccan authorities to help
the policy makers and public health administration to make
suitable decisions in the near future.	
\end{abstract}



\maketitle


\section{Introduction}

Coronavirus disease 2019 (COVID-19), reclassified as a pandemic
by the World Health Organization (WHO) on March 11, 2020 \cite{MHM},
is an infectious disease caused by a new type of virus belonging
to the coronaviruses family and recently named severe acute
respiratory syndrome coronavirus 2 (SARS-CoV-2) \cite{ICTW}.
All the countries affected by this disease have taken many preventive measures,
including containment. The containment established by the Moroccan government
and the public authorities at the right time made it possible to avoid the worst:
according to the minister of Health, at least 6000 lives were saved thanks
to the measures adopted to face the spread of this pandemic \cite{url:HIVdata:morocco}.
The resistance measures, regarded as necessary and urgent, cannot be sustainable.

Actually, the deconfinement is a new stage entered by the COVID-19 pandemic.
Therefore, several countries strategically planned their deconfinement strategies.
The extension of the state of emergency in Morocco until May 20, 2020
will no doubt have economic repercussions. If Morocco won the first round,
or at least limited the consequences, especially in terms of limiting
the pandemic and health management of the situation, the second seems
difficult and complex. Indeed, it must not only be well thought out but
also its axes of resistance have to be well-identified. In this context,
all efforts should be focused on stabilizing the economy by intelligently
relying on resources. Economic deconfinement is part of the solution
and should be gradual and concerted. Indeed, it is absurd to think that
the return to the normality is in the near months, because the unavailability
of an effective vaccine implies that the virus will always be with us
in the near future, which poses a risk for the population. This economic
deconfinement should be prepared and accompanied by other related measures, 
in particular under health, security, education and social assistance.
In this period of general crisis, the response must try to mitigate the
impacts on priority sectors, such as agriculture, agrifood, transport
and foreign trade, in relation to imports that are vital to the Moroccan economy.
The challenge is to ensure resistance and a continuity of value creation
while preventing a sector from being detached from the economic body.
So to speak, priority must be given to vital sectors whose health
directly affects all Moroccan activity, while protecting those bordering
on chaos. According to the deconfinement strategy,
which is applied by the Moroccan authorities, it is mandatory to
study the occurrence of an eventual second wave and it's magnitude.

Mathematical modeling through dynamical systems plays an important
role to predict the evolution of COVID-19 transmission \cite{MR4095770,MR4093642}.
However, while taking into account the deconfinement policies, the environmental
effects and the social fluctuations should not be neglected
in such a mathematical study in order to describe well the dynamics
and consider an additional degree of realism \cite{Tang,Wu,Kuniya,Fanelli,Boudrioua}.
For these reasons, we describe here the dynamics of the deconfinement strategy
by a new D-COVID-19 model, governed by delayed stochastic
differential equations (DSDE), as follows:
\begin{equation}
\label{Sys1}
\begin{cases}
dS(t)=\left(\rho C(t)-\delta S(t)-\beta(1-u)\dfrac{S(t)I_{s}(t)}{N}\right)dt
-\sigma_{1}(1-u)\dfrac{S(t)I_{s}(t)}{N}dB_{1}(t)+ \sigma_{2}(C(t)-S(t))dB_{2}(t),\\
dC(t)=\left(\delta S(t)- \rho C(t)\right)dt + \sigma_{2}(S(t)-C(t))dB_{2}(t), \\
dI_{s}(t)=\left(\beta\epsilon (1-u)\dfrac{S(t-\tau_1)I_{s}(t-\tau_1)}{N}-\alpha I_{s}(t)
-(1-\alpha)(\mu_s+\eta_{s})I_s(t)\right)dt \\
\qquad \qquad +\sigma_{1}\left(\epsilon (1-u)\dfrac{S(t-\tau_1)I_{s}(t-\tau_1)}{N}\right)dB_{1}(t)
+\sigma_{3}(\mu_s+\eta_{s}-1)I_s(t)dB_{3}(t),\\
dI_{a}(t)=\left(\beta(1-\epsilon)(1-u)\dfrac{S(t-\tau_1)I_{s}(t-\tau_1)}{N}
-\eta_{a}I_{a}(t)\right)dt
+\sigma_{1}(1-\epsilon)(1-u)\dfrac{S(t-\tau_1)I_{s}(t-\tau_1)}{N}dB_{1}(t),\\
dF_{b}(t)=\bigg(\alpha\gamma_{b}I_{s}(t-\tau_2)-\big(\mu_b+r_b\big)F_{b}(t)\bigg)dt
+\sigma_{3}\gamma_{b}I_{s}(t-\tau_2)dB_{3}(t),\\
dF_{g}(t)=\bigg(\alpha\gamma_{g}I_{s}(t-\tau_2)-\big(\mu_g+r_g\big)F_{g}(t)\bigg)dt
+\sigma_{3}\gamma_{g}I_{s}(t-\tau_2)dB_{3}(t),\\
dF_{c}(t)=\bigg(\alpha\gamma_{c}I_{s}(t-\tau_2)-\big(\mu_c+r_c\big)F_{c}(t)\bigg)dt
+\sigma_{3}\gamma_{c}I_{s}(t-\tau_2)dB_{3}(t),\\
dR(t)=\bigg(\eta_{s}(1-\alpha)I_s(t-\tau_3)+\eta_{a}I_{a}(t-\tau_3)
+r_{b}F_{b}(t-\tau_4)+r_g F_{g}(t-\tau_4)
+r_c F_{c}(t-\tau_4)\bigg)dt \\
\qquad \qquad -\sigma_3\eta_{s}I_{s}(t-\tau_3)dB_3(t),\\
dM(t)=\bigg(\mu_s(1-\alpha)I_s(t-\tau_3)+\mu_bF_{b}(t-\tau_4)
+\mu_gF_{g}(t-\tau_4)+\mu_cF_{c}(t-\tau_4)\bigg)dt
-\sigma_{3}\mu_{s}I_{s}(t-\tau_3)dB_{3}(t),
\end{cases}
\end{equation}
where $S$ represents the susceptible sub-population, which is not infected
and has not been infected before but is susceptible to develop the disease
if exposed to the virus; $C$ is the confined sub-population;
$I_s$ is the symptomatic infected sub-population, which has not yet been treated,
it transmits the disease, and outside of proper support it can progress
to spontaneous recovery or death; $I_a$ is the asymptomatic infected sub-population
who is infected but does not transmit the disease, is not known by the health
system and progresses spontaneously to recovery; $F_b$, $F_g$ and $F_c$
are the patients diagnosed, supported by the Moroccan health system
and under quarantine, and subdivided into three categories: benign, severe,
critical forms, respectively. Finally, $R$ and $M$ are the recovered
and died classes, respectively. At each instant of time, the equation
\begin{equation*}
D(t):=\mu_s(1-\alpha)I_s(t-\tau_3)+\mu_bF_{b}(t-\tau_4)
+\mu_gF_{g}(t-\tau_4)+\mu_cF_{c}(t-\tau_4)
-\sigma_3\mu_{s}I_{s}(t-\tau_3)\dfrac{\Delta B_3(t)}{\Delta t}
=\dfrac{\Delta M(t)}{\Delta t}
\end{equation*}
gives the number of the new dead due to disease. The parameter $1-u$
represents the level of measures undertaken on the susceptible population
while $\delta$ is the confinement rate and $\rho$ represents the deconfinement rate.
We adopt the bilinear incidence rate to describe the infection of the disease
and use the parameter $\beta$ to denote the transmission rate. It is reasonable
to assume that the infected individuals are subdivided into individuals
with symptoms and others without symptoms, for which we employ the parameter
$\epsilon$ to denote the proportion for the symptomatic individuals
and $1-\epsilon$ for the asymptomatic ones. The parameter $\alpha$ measures
the efficiency of public health administration for hospitalization. Diagnosed
symptomatic infected population is completely distributed into one of the 
three forms $F_b$, $F_g$ and $F_c$, by the rates $\gamma_b$, $\gamma_g$ and $\gamma_c$, 
respectively. Then, $\gamma_b+\gamma_g+\gamma_c= 1$. The mean recovery period 
of these forms are denoted by  $1/r_b$, $1/r_g$ and $1/r_c$, respectively. 
The latter forms die also with the rates $\mu_b$, $\mu_g$ and $\mu_c$,
respectively. Symptomatic infected population, which is not diagnosed, moves
to the recovery compartment with a rate $\eta_s$ or dies with a rate $\mu_s$.
On the other hand, asymptomatic infected population moves to the recovery 
compartment with a rate $\eta_a$. The time delays $\tau_1$ and $\tau_2$ denote
the incubation period and the period of time needed before the charge by the health 
system, respectively. The time delays $\tau_3$ and $\tau_4$ denote the time 
required before the death of individuals coming from the compartments $I_s$ 
and the three forms $F_b$, $F_g$ and $F_c$, respectively. Here, $B_{1}(t)$, 
$B_{2}(t)$ and $B_{3}(t)$ are independent standard Brownian motions
defined on a complete probability space $(\Omega,\mathcal{F},\mathbb{P})$
with a filtration $\{\mathcal{F}_{t}\}_{t\geq0}$ and satisfying the usual conditions,
that is, they are increasing and right continuous while $\mathcal{F}_{0}$ 
contains all P-null sets and $\sigma_{i}$ represents the intensity 
of $B_{i}$, $i=1,2,3$. The schematic diagram of our extended model 
is illustrated in Figure~\ref{fig:sd}.
\begin{center}
\includegraphics[scale=1.4]{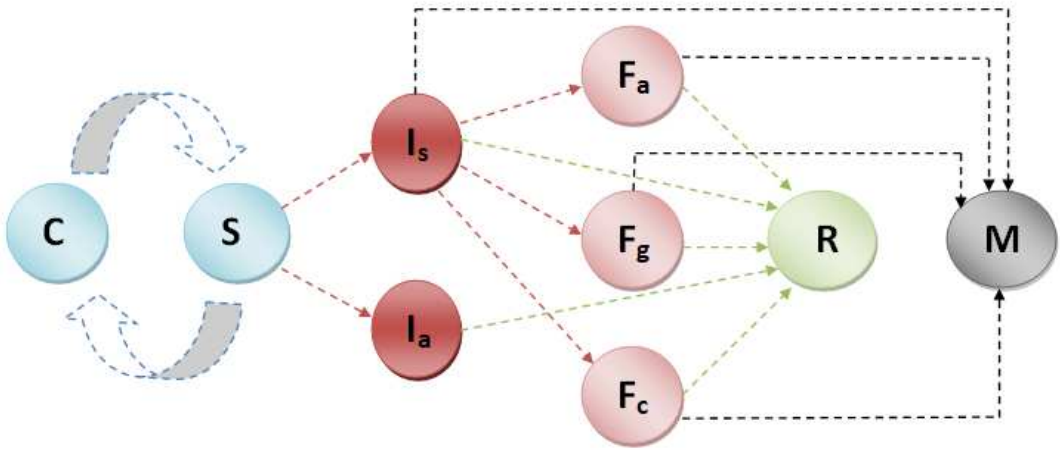}
\captionof{figure}{Schematic diagram of model \eqref{Sys1}.}
\label{fig:sd}
\end{center}

\begin{remark}
The stochasticity is introduced in model \eqref{Sys1} by perturbing 
the most sensitive parameters: $\beta$, $\alpha$, $\delta$, and $\rho$.
\end{remark}

\begin{remark}
For the sake of simplicity, we have assumed that the parameters 
$\delta$ and $\rho$ are perturbed with the same intensities, that is,
we assume that Moroccan individuals possess the same behaviors and reactions 
towards the authorities instructions.
\end{remark}

\begin{remark}
Note that the multipliers of $I_s(t-\tau_2)$ terms are the same as $\sigma_3dB_3(t)$ 
although they are premultiplied by different constants.
Indeed, the portion of diagnosed symptomatic infected population 
is completely distributed into the three forms $F_b$, $F_g$, and $F_c$, 
by the rates $\gamma_b$, $\gamma_g$ and $\gamma_c$, respectively. 
Then, $\gamma_b+\gamma_g+\gamma_c= 1$. In addition, we assume 
that the parameter $\alpha$, which measures the efficiency of public 
health administration for hospitalization, undergoes random fluctuations. 
\end{remark}

\begin{remark}
For illustration and clarification purposes, let us suppose, as an example, 
that the disease progression is started from 15th of March  
and value of $\tau_1$ is 5. Then a susceptible individual, 
after contact with an infected one at instant $t$, becomes himself 
infected at instant $t+\tau_1$. Suddenly, the compartment 
of the infected is fed at the instant $t$ by the susceptible 
infected at the instant $t-\tau_1$. Therefore, 
in the considered situation,
when the infection starts at March 15, the term 
$\beta\epsilon(1-u)S(T)I_s(T)/N$ of new infected 
is equal to zero on 12th, 13th and 14th March,
due to the absence of the infection.
\end{remark}

\begin{remark}
Temporarily asymptomatic individuals are included in the class $I_s$ of symptomatic, 
while individuals in $I_a$, who are permanently asymptomatic, will remain asymptomatic 
until recovery and will not spread the virus, a fact which has been recently 
confirmed by the World Health Organization.
\end{remark}

For biological reasons, we assume that the initial conditions
of system (\ref{Sys1}) satisfy:
\begin{equation}
\label{ic}
\begin{array}{ll}
S(\theta) =\phi_{1}(\theta)\geq 0,
\quad C(\theta) =\phi_{2}(\theta)\geq 0,
\quad I_a(\theta)=\phi_{3}(\theta)\geq 0,\\
I_s(\theta)=\phi_{4}(\theta)\geq 0,
\quad F_b(\theta) =\phi_{5}(\theta)\geq 0,
\quad F_g(\theta)=\phi_{6}(\theta)\geq 0,\\
F_c(\theta)=\phi_{7}(\theta) \geq 0,
\quad R(\theta)=\phi_{8}(\theta)\geq 0,
\quad M(\theta)=\phi_{9}(\theta)\geq 0,
\end{array}
\end{equation}
where $\theta \in [-\tau,0]$ and
$\tau=\max\{\tau_1,\tau_2\,\tau_3,\tau_4\}$.

The rest of the paper is organized as follows.
Section~\ref{sec:2} deals with the existence and uniqueness
of a positive global solution that ensures the well-posedness
of the D-COVID-19 model \eqref{Sys1}. A sufficient condition
for the extinction is established in Section~\ref{sec:3}.
Then, some numerical scenarios, to assess the effectiveness of the
adopted deconfinement strategy, are presented in Section~\ref{sec:4}.
The paper ends up with Section~\ref{sec:5} of conclusion.


\section{Existence and uniqueness of a positive global solution}
\label{sec:2}

Let us denote $\mathbb{R}^9_+
:=\{(x_1,x_2,x_3,x_4,x_5,x_6,x_7,x_8,x_9)\mid x_i>0, \, i=1,2,\ldots,9\}$.
We begin by proving the following result.

\begin{theorem}
\label{thm:2.1}
For any initial value satisfying condition \eqref{ic}, there is a unique solution
$$
x(t)=(S(t),C(t),I_s(t),I_a(t),F_b(t),F_g(t),F_c(t),R(t),M(t))
$$
to the D-COVID-19 model \eqref{Sys1} that remains in $\mathbb{R}^9_+$ with probability one.
\end{theorem}

\begin{proof}
Since the coefficients of the Stochastic Differential Equations with several delays
\eqref{Sys1} are locally Lipschitz continuous, it follows from \cite{Mao} that for
any square integrable initial value $x(0)\in\mathbb{R}^9_+$, which is independent
of the considered standard Brownian motion $B$, there exists a unique local
solution $x(t)$ on $t\in[0,\tau_e)$, where  $\tau_e$ is the explosion time.
For showing that this solution is global, knowing that the linear growth condition
is not verified, we need to prove that $\tau_e=\infty$. Let $k_0>0$ be sufficiently
large for $\dfrac{1}{k_0}<x(0)<k_0$. For each integer $k\geq k_0$, we define
the stopping time
$\tau_k := \inf\left\lbrace t\in[0,\tau_e)/x_i(t)\notin
\left( \dfrac{1}{k},k\right)\; \text{for some}\; i=1,2,3\right\rbrace$,
where $\inf \emptyset=\infty$. It is evident that $\tau_k\leq \tau_e$.
Let $T>0$, and define the twice differentiable function
$V$ on $\mathbb{R}^3_+\rightarrow \mathbb{R}^+$ as follows:
$$
V(x):=(x_1+x_2+x_3)^2+\frac{1}{x_1}+\frac{1}{x_2}+\frac{1}{x_3}.
$$
By It\^{o}'s formula, for any 
$0\leq t\leq \tau_k \wedge T$ and $k\geq 1 $ we have
$$
dV(x(t))=LV(x(t))dt+\sigma(x(t))dB_t,
$$
where $L$ is the differential operator of function $V$:
\begin{equation*}
\begin{split}
L&V(x(t)) = \left( 2(S(t)+C(t)+I_s(t))-\dfrac{1}{S^2(t)}\right)
\left( \rho C(t)-\delta S(t)-\beta (1-u)\dfrac{S(t) I_s(t )}{N}\right)\\
&+ \dfrac{1}{2}\left( 2+\frac{2}{S^3(t)}\right) \left( \left(
-\sigma_1(1-u)\dfrac{S(t) I_s(t )}{N}\right)^2
+(\sigma_2(C(t)-S(t)))^2\right) \\
&+ \left( 2(S(t)+C(t)+I_s(t))-\dfrac{1}{C^2(t)}\right) (\delta S(t)
-\rho C(t))+\dfrac{1}{2}\left( 2+\frac{2}{C^3(t)}\right) (-\sigma_2 (C(t)-S(t)))^2\\
&+ \dfrac{1}{2}\left( 2+\frac{2}{C^3(t)})(-\sigma_2(C(t)-S(t))\right)^2\\
&+ \left( 2(S(t)+C(t)+I_s(t))-\frac{1}{I^2_s(t)}\right) \left(
\beta \epsilon (1-u)\frac{S(t-\tau_1) I_s(t-\tau_1 )}{N}-\alpha I_s(t)
-(1-\alpha)(\mu_s+\eta_s)I_s(t)\right) \\
&+ \dfrac{1}{2}\left( 2+\dfrac{2}{I^3_s(t)}\right) \left( \left(
-\sigma_3 I_s(t)(1-\mu_s-\eta_s)\right) ^2+\left( \sigma_1
\epsilon (1-u)\frac{S(t-\tau_1) I_s(t-\tau_1 )}{N}\right) ^2\right).
\end{split}
\end{equation*}
Thus,
\begin{eqnarray*}
LV(x(t))
&\leq & 2(S(t)+C(t)+I_s(t))\rho C(t)+\frac{\delta}{S(t)}
+\frac{\beta (1-u)S(t) I_s(t)}{NS^{2}(t)}\\
&+&  \left( 1+\dfrac{1}{S^3(t)}\right) \left( \left(
\sigma_1(1-u)\dfrac{S(t) I_s(t)}{N}\right)^2
+ (\sigma_2(C(t)-S(t)))^2\right) \\
&+&  2(S(t)+C(t)+I_s(t))\delta S(t)+\dfrac{\rho}{S(t)}
+\left( 1+\dfrac{1}{C^3(t)}\right) (\sigma_2(C(t)-S(t)))^2\\
&+& 2(S(t)+C(t)+I_s(t))\beta \epsilon (1-u)\frac{S(t-\tau_1) I_s(t-\tau_1 )}{N}\\
&+& \frac{\alpha}{I_s(t)}+(1-\alpha)(\mu_s+\eta_s)\frac{1}{I_s(t)}\\
&+& \left( 1+\dfrac{1}{I^3_s(t)}\right) \left( \left(
-\sigma_3 I_s(t)(1-\mu_s-\eta_s\right) ^2+\left( \sigma_1
\epsilon (1-u)\dfrac{S(t-\tau_1) I_s(t-\tau_1 )}{N}\right) ^2\right).
\end{eqnarray*}
By applying the elementary inequality $2ab\leq a^2+b^2$, we can 
easily increase the right-hand side of the previous inequality
to obtain that
$$
LV(x)\leq D(1+V(x)),
$$
where $D$ is an adequate selected positive constant. 
By integrating both sides of the equality
$$
dV(x(t))=LV(x(t))dt+\sigma(x(t))dB_t
$$
between $t_0$ and $t\wedge \tau_k$ and acting the expectation,
which eliminates the martingale part, we get that
\begin{eqnarray*}
E(V(x(t\wedge \tau_k))
&=& E(V(x_0))+E\int^{t\wedge \tau_k}_{t_0}LV(x_s))ds\\
&\leq & E(V(x_0))+E\int^{t\wedge \tau_k}_{t_0}D(1+V(x_s))ds\\
&\leq &  E(V(x_0))+DT+\int^{t \wedge \tau_k}_{t_0}EV(x_s))ds.
\end{eqnarray*}
Gronwall's inequality implies that
$$
E(V(x(t\wedge \tau_k))\leq (EV(x_0)+DT)\exp(CT).
$$
For $\omega \in \{\tau_k\leq T\}$, $x_i(\tau_k)$
equals $k$ or $\dfrac{1}{k}$ for some $ i=1,2,3$.
Hence,
$$
V(x_i(\tau_k))\geq \left(k^2+\frac{1}{k}\right)
\wedge \left(\frac{1}{k^2}+k\right).
$$
It follows that
\begin{eqnarray*}
(EV(x_0)+DT)\exp(CT)
&\geq & E\left( \chi_{\{\tau_k \leq T\}}(\omega) V(x_{\tau_k})\right)\\
&\geq & \left(k^2+\frac{1}{k}\right)
\wedge \left(\frac{1}{k^2}+k\right)P\left(\tau_k \leq T\right).
\end{eqnarray*}
Letting $k\rightarrow \infty $, we get $P(\tau_e \leq T)=0$.
Since $T$ is arbitrary, we obtain $P(\tau_e =\infty)=1$.
With the same technique, we also deduce that the rest 
of the variables of the system are positive on $[0,\infty)$. 
This concludes the proof.
\end{proof}


\section{Extinction of the disease}
\label{sec:3}

In this section, we obtain a sufficient condition
for the extinction of the disease.

\begin{theorem}
\label{thm:3.1}
Let $(S(t),C(t),I_s(t),I_a(t),F_b(t),F_g(t),F_c(t),R(t),M(t))$
be a solution of the D-COVID-19 model \eqref{Sys1} with positive
initial value defined in \eqref{ic}. Assume that
$$
\sigma_1^2> \dfrac{\beta^2}{2(\alpha+(1-\alpha)(\mu_s+\eta_s))}.
$$
Then,
\begin{equation*}
\limsup_{t\to \infty} \ln\dfrac{I_s(t)}{t}<0.
\end{equation*}
Namely, $I_s(t)$ tends to zero exponentially a.s.,
that is, the disease dies out with probability $1$.
\end{theorem}

\begin{proof}
Let
\begin{align*}
d\ln I_s(t)
&=  \left[ \dfrac{1}{I_s(t)}\left( \dfrac{\beta \epsilon (1-u)
S(t-\tau_1)I_s(t-\tau_1)}{N} - \alpha I_s(t) - (1-\alpha)(\mu_s+\eta_s)I_s(t)\right) \right.\\
& \left. -\dfrac{1}{2I_s^{2}(t)}\left( \left(\sigma_1 \dfrac{\beta
\epsilon (1-u)S(t-\tau_1)I_s(t-\tau_1)}{N} \right)^{2}
+ \left( \sigma_3(\mu_s+\eta_s-1)I_s(t)\right)^{2}\right)\right] dt\\
&+ \dfrac{1}{I_s(t)} \sigma_1 \dfrac{\beta \epsilon (1-u)S(t-\tau_1)
I_s(t-\tau_1)}{N} dB_1 + \dfrac{1}{I_s(t)} \sigma_3 (\mu_s+\eta_s-1)I_s(t) dB_3.
\end{align*}
To simplify, we set
\begin{align*}
G(t)&:= \dfrac{ \epsilon (1-u)S(t-\tau_1)I_s(t-\tau_1)}{N},\\
R_{1}(t)&:= \dfrac{\sigma_1}{I_s(t)}  \dfrac{\beta \epsilon (1-u)
S(t-\tau_1)I_s(t-\tau_1)}{N}  =\dfrac{\beta \sigma_1}{I_s(t)} G,\\
R_{3}(t)&:= \dfrac{\sigma_3}{I_s(t)}  (\mu_s+\eta_s-1)I_s(t),\\
H &:=-\alpha -(1-\alpha)(\mu_s+\eta_s).
\end{align*}
We then get
\begin{align*}
d\ln I_s(t) &=
\dfrac{\beta G(t)}{I_s(t)} - H(t) -\dfrac{1}{2} \left( \left(
\dfrac{\sigma_1 G(t)}{I_s(t)} \right)^{2}+ \left( \sigma_3(\mu_s
+\eta_s-1)I_s(t)\right)^{2}\right)+ R_1(t) dB_{1} + R_3(t) dB_{3}\\
&= -\dfrac{\sigma_1^{2}}{2}\left[ \left(
\dfrac{G(t)}{I_s(t)}\right) ^{2}- \dfrac{2\beta}{\sigma_1^{2}}
\dfrac{G(t)}{I_s(t)}\right] + H + R_1(t) dB_{1}+ R_3(t) dB_{3}\\
&= -\dfrac{\sigma_1^{2}}{2} \left[ \left( \dfrac{G(t)}{I_s(t)}
- \dfrac{\beta}{\sigma_1^{2}} \right) ^{2}
- \dfrac{\beta^{2}}{\sigma_1^4} \right] + H + R_1(t) dB_{1} + R_3(t) dB_{3}\\
&\leq \,\dfrac{\beta^{2}}{2\sigma_1^{2}}  + H + R_1(t) dB_{1} + R_3(t) dB_{3}.
\end{align*}
Hence,
\begin{align*}
\dfrac{\ln I_s(t)}{t}
&\leq \dfrac{\ln I_s(0)}{t} + \dfrac{\beta^{2}}{2\sigma_1^{2}}
+ H + \dfrac{M_1(t)}{t} + \dfrac{M_3(t)}{t} ,
\end{align*}
where
\begin{equation*}
M_1(t)= \int_{0}^{t}R_{1}(s)dB_{1},
\quad M_3(t)= \int_{0}^{t}R_{3}(s)dB_{3}.
\end{equation*}
We have
\begin{eqnarray*}
\left<M_1,M_1\right>_t
&=& \int^t_0\left( \frac{1}{I_s(s)}\sigma_1G(s)\right)^2ds\\
&=& \int^t_0{\sigma_1}^2\epsilon^2(1-u)^2\dfrac{S(t-\tau_1)^2I_s(t-\tau_1)^2}{N^2I_s^2}ds\\
&\leq &  \int^t_0{\sigma_1}^2\epsilon^2(1-u)^2\frac{N^4}{N^2}\frac{1}{I^2_s}ds\\
&\leq &  \int^t_0{\sigma_1}^2\epsilon^2(1-u)^2N^2ds.
\end{eqnarray*}
Then,
$$
\underset{t\rightarrow\infty}{\limsup}\dfrac{<M_1,M_1>_t}{t}
\leq {\sigma_1}^2\epsilon^2(1-u)^2N^2<\infty.
$$
From the large number theorem for martingales \cite{Grai}, we deduce that
$$
\underset{t\rightarrow\infty}{\lim}\frac{M_1(t)}{t}=0.
$$
We also have
\begin{eqnarray*}
<M_2,M_2>_t
&=& \int^t_0\left( \frac{1}{I_s(s)}\sigma_3(\mu_s+\eta_s-1)I_s(s)\right)^2ds\\
&=& \int^t_0\left( \sigma^2_3(\mu_s+\eta_s-1)^2\right)ds\\
&=& \left( \sigma^2_3(\mu_s+\eta_s-1)^2\right)t.
\end{eqnarray*}
Thus,
$$
\underset{t\rightarrow\infty}{\limsup}\dfrac{<M_2,M_2>_t}{t}
\leq \sigma^2_3(\mu_s+\eta_s-1)<\infty.
$$
We deduce that
$$
\underset{t\rightarrow\infty}{\lim}\frac{M_2(t)}{t}=0.
$$
Subsequently,
$$
\underset{t\rightarrow\infty}{\limsup}\dfrac{\ln I_s(t)}{t}
\leq \frac{\beta^2}{2\sigma_1^2}-\alpha-(1-\alpha)(\mu_s+\eta_s).
$$
We conclude that if $\dfrac{\beta^2}{2\sigma_1^2}
-\alpha-(1-\alpha)(\mu_s+\eta_s)<0$,
then $\underset{t\rightarrow\infty}{\lim I(t)}=0$.
The proof is complete.
\end{proof}


\section{Results and discussion}
\label{sec:4}

In this section, we simulate the forecasts of the D-COVID-19 model
\eqref{Sys1}, relating the deconfinement strategy adopted 
by Moroccan authorities with two scenarios. 
We assume $u$ defined as follows:
$$
u=\left\{
\begin{array}{ll}
u_{0}, & \hbox{$ \text{on}\, [\text{March 2},\text{March 10]}$;} \\
u_{1}, & \hbox{$ \text{on}\, (\text{March 10},\text{March 16]}$;}\\
u_{2}, & \hbox{$ \text{on}\, (\text{March 16},\text{March 20]}$;}\\
u_{3}, & \hbox{$ \text{on}\, (\text{March 20},\text{April 6]}$;}\\
u_{4}, & \hbox{$ \text{on}\, (\text{April 6},\text{April 25]}$;}\\
u_{5}, & \hbox{$ \text{from}\, \text{April 25} \text{ on}$;}\\
\end{array}
\right.
$$
where $u_{i}\in (0,1]$, for $i=0,1,2,3,4,5$, measures
the effectiveness of applying the multiple preventive interventions
imposed by the authorities and presented in Table~\ref{inter}.
\begin{center}
\captionof{table}{Summary of considered non-pharmaceutical interventions.}
\label{inter}
\begin{tabular}{ll}
\hline \hline
\textbf{Policies} &  \textbf{Decisions made at the government level} \\ \hline
With minimal social distancing  measure &  $ u = 0.1,$
\hbox{$ \text{on}\, [\text{March 2},\text{March 16]}$}  and\\
 & $ u= 0.2$,  \hbox{$ \text{on}\,(\text{March 16},\text{March 20]}$}  \\ \hline
With middle social distancing measure & $ u= 0.1$,
\hbox{$\text{on}\, [\text{March 2},\text{March 16]}$}, \\
& $ u= 0.2$,  \hbox{$\text{on}\,(\text{March 16},\text{March 20]}$} and \\
& $ u= 0.4$, \hbox{$\text{on}\, (\text{March 20},\text{April 6]}$} \\
\hline With high social distancing measure &    $ u= 0.1,
\hbox{$\text{on}\, [\text{March 2},\text{March 16]}$} $,  \\
 &  $ u= 0.2$,  \hbox{$\text{on}\, (\text{March 16},\text{March 20]}$},  \\
& $ u= 0.4$, \hbox{$\text{on}\, (\text{March 20},\text{April 6]}$} and   \\
& $ u= 0.6$, \hbox{$\text{on}\, (\text{April 6},\text{April 25]}$} \\ \hline
With maximal social distancing measure
&    $ u= 0.1,  \hbox{$\text{on}\, [\text{March 2},\text{March 16]}$} $,  \\
 &  $ u= 0.2$,  \hbox{$\text{on}\, (\text{March 16},\text{March 20]}$},  \\
& $ u= 0.4$, \hbox{$\text{on}\, (\text{March 20},\text{April 6]}$},   \\
& $ u= 0.6$, \hbox{$\text{on}\, (\text{April 6},\text{April 25]}$} and \\
& $ u= 0.7$, \hbox{$\text{from}\, \text{April 25}$ \text{on.}} \\ \hline \hline
\end{tabular}
\end{center}

COVID-19 is known as a highly contagious disease and its transmission rate,
$\beta$, varies from country to country, according to the density
of the country and movements of its population. Ozair et al. \cite{Ozair}
assumed $\beta$ to be $ [0.198-0.594]$ per day for Romania,
and $ [0.097-0.291] $ per day for Pakistan. Further, Kuniya \cite{Kuniya}
estimated $\beta$ as $0.26$ $(95\%CI,\, 2.4-2.8)$. Observing
the number of daily reported cases of COVID-19 in Morocco,
we estimate $\beta$ as $0.4517$ $(95\%CI,\, 0.4484-0.455)$.
After the infection, the patient remains in a latent period for $5.5$ days \cite{WHO1,Stephen},
in average, before becoming symptomatic and infectious or asymptomatic with a percentage
that varies from $20.6 \%$ of infected population to $39.9 \%$ \cite{Mizumoto},
while the time needed before his hospitalization is estimated to be $7.5$ days \cite{Huang,Wang,Haut}.
All the parameter values chosen for the D-COVID-19 model \eqref{Sys1} 
are summarized in Table~\ref{values1}.
\begin{center}
\captionof{table}{Parameter values for the D-COVID-19 model (\ref{Sys1}).}
\label{values1}
\begin{tabular}{lcccccccccc} \hline \hline
Parameter &   $\beta$& $\epsilon$ &$\gamma_{b}$&$\gamma_{g}$
&$\gamma_{c}$&$\alpha$&$\eta_{a}$&$\eta_{s}$  &$\mu_{s}$ \\ \hline
Value & $0.4517$ & $0.794$& $0.8$&$0.15$&$0.05$&$  0.06 $
&$1/21$&$0.8/21$&$0.01/21$\\ \hline \hline\\[0.3cm]
\end{tabular}
\begin{tabular}{lccccccccccc} \hline \hline
Parameter &$\mu_{b}$& $\mu_{g}$&  $\mu_{c}$ & $r_{b}$ & $r_{g}$ & $r_{c}$
& $ \tau_1 $ &$ \tau_2 $&$ \tau_3 $& $ \tau_4 $\\ \hline
Value &$ 0 $& $ 0 $&$0.4/13.5$ & $ 1/13.5$ & $1/13.5$&$0.6/13.5$
&  $ 5.5$& $7.5$ & $ 21$ & $13.5$\\ \hline \hline
\end{tabular}
\end{center}

\begin{remark}
From a biological point of view, the latency period 
is independent of the region or country under study, 
depending only on the structural nature 
of the SARS-CoV-2 coronavirus. 
\end{remark}

We consider that all measures and the adopted confinement strategy 
previously discussed are conserved. The evolution on the number of 
diagnosed infected positive individuals given by the D-COVID-19 model \eqref{Sys1}
versus the daily reported confirmed cases of COVID-19 in Morocco, 
from March 2 to May 6, is presented in Figure~\ref{Historique}. 
We see that the curve generated by the D-COVID-19 model \eqref{Sys1} 
follows the trend of the daily reported cases in Morocco. So, we confirm 
that the implemented measures taken by the authorities have an explicit 
impact on the propagation of the virus in the population since the curve 
of the D-COVID-19 model \eqref{Sys1} has been flattening from April 17 
and tends to go towards the extinction of the disease from May 05. 
In Figure~\ref{Evolution}, we see that Morocco has spent almost $40\%$  
of the total duration of the epidemic at May 11 and will reach extinction 
after four months, in average, from the start of the epidemic on March 2, 2020 ($t=0$).
\begin{center}
\includegraphics[width=15cm,height=7cm]{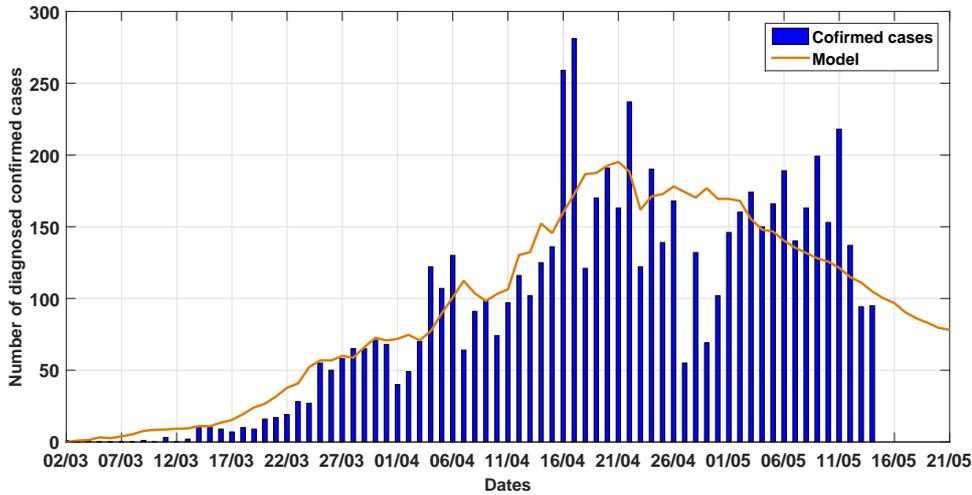}
\captionof{figure}{Evolution of COVID-19 confirmed cases in Morocco per day:
curve predicted by our model \eqref{Sys1} accordigly with 
Tables~\ref{inter} and \ref{values1} versus real data.}
\label{Historique}
\end{center}
\begin{center}
\includegraphics[width=15cm,height=7cm]{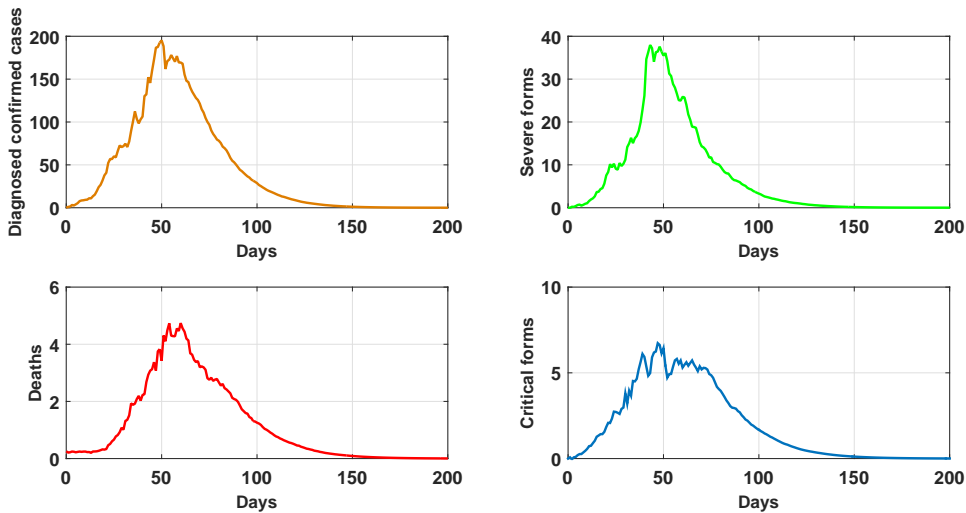}
\captionof{figure}{Evolution given by the D-COVID-19 model \eqref{Sys1}
without deconfinement ($\rho=0)$.}
\label{Evolution}
\end{center}

To prove the biological importance of delay parameters, we give the
graphical results of Figure~\ref{Delays}, which allow to compare 
the evolution of diagnosed positive cases with and without delays.
\begin{figure}[ht!]
\centering
\includegraphics[width=15cm,height=7cm]{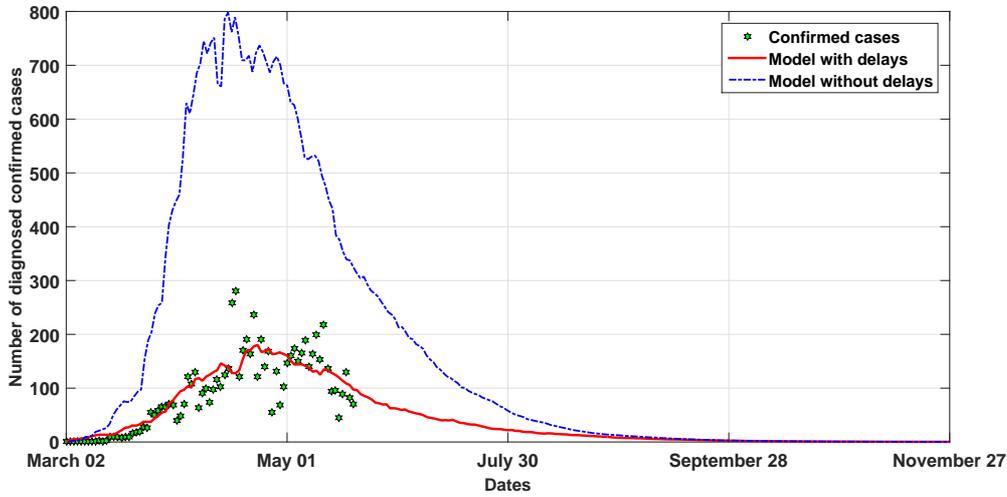}
\captionof{figure}{Effect of delays on the diagnosed confirmed cases
versus clinical data.}
\label{Delays}
\end{figure}
We observe in Figure~\ref{Delays} a high impact of delays on the number
of diagnosed positive cases. Indeed, the plot of model (\ref{Sys1})
without delays $(\tau_i=0,\ i=1,2,3,4$) is very far from the clinical data.
Thus, we conclude that delays play an important role in the study
of the dynamical behavior of COVID-19 worldwide, especially in Morocco,
and allows to better understand the reality.

In Figures~\ref{Deconfinement001}, \ref{Deconfinement01}, and \ref{Deconfinement015},
we consider the deconfinement of $30\%$ of the population returning to work from May 20,
and this proportion is immediately integrated into the susceptible population.
Numerical simulations are presented for three possible scenarios. In the first,
we consider that the whole population highly respects the majority of the measures
announced by the authorities in relation with the deconfinement (Figure~\ref{Deconfinement001}).
The second and third scenarios show the direct impact on the curves when the population
moderately respects the measures with different levels, $\sigma_2=0.10$ and $\sigma_2=0.15$,
respectively. With the last two scenarios we observe the growth in the final number of infected,
deaths, severe and critical forms, which are the most important to monitor,
since the health system should not be saturated. It is also important to note the appearance
of a second significant peak and the fact that the time required for extinction becomes longer,
which relates to the value of $ \sigma_2 $ (Figures~\ref{Deconfinement01} and \ref{Deconfinement015}).
\begin{center}
\includegraphics[width=15cm,height=7cm]{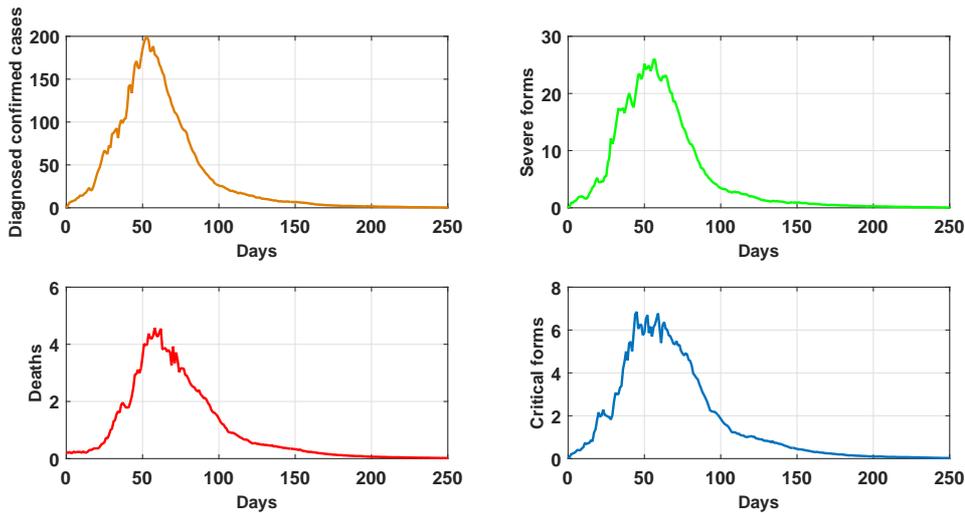}
\captionof{figure}{Evolution of the D-COVID-19 model \eqref{Sys1}
with deconfinement ($\rho=0.3)$ from May 20, 2020
and high effectiveness of the measures ($ \sigma_2=0.01$).}
\label{Deconfinement001}
\end{center}
\begin{center}
\includegraphics[width=15cm,height=7cm]{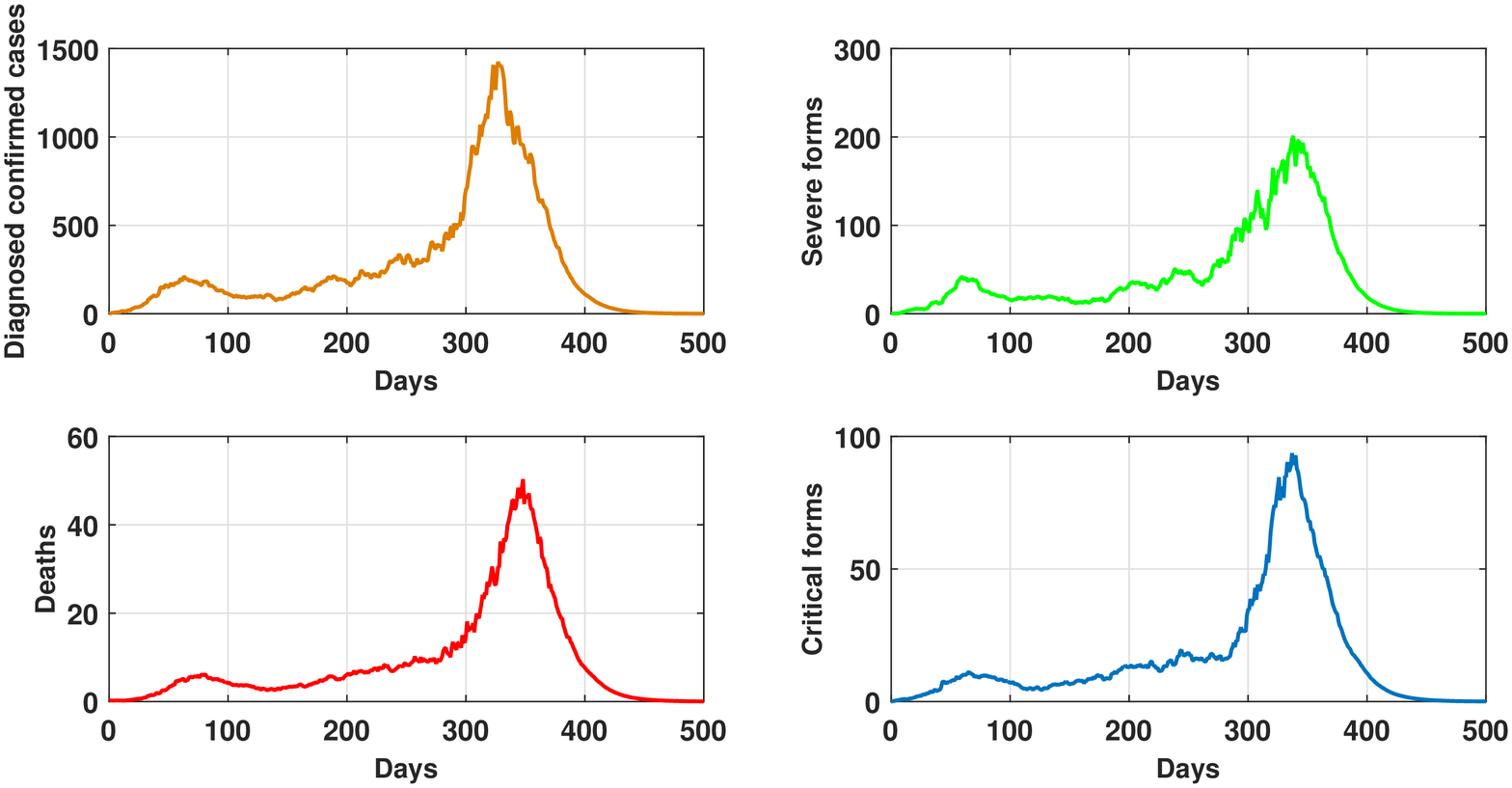}
\captionof{figure}{Evolution of the D-COVID-19 model \eqref{Sys1}
with deconfinement ($\rho=0.3) $ from May 20, 2020 
and moderate effectiveness of the measures ($\sigma_2=0.10$).}
\label{Deconfinement01}
\end{center}
\begin{center}
\includegraphics[width=15cm,height=7cm]{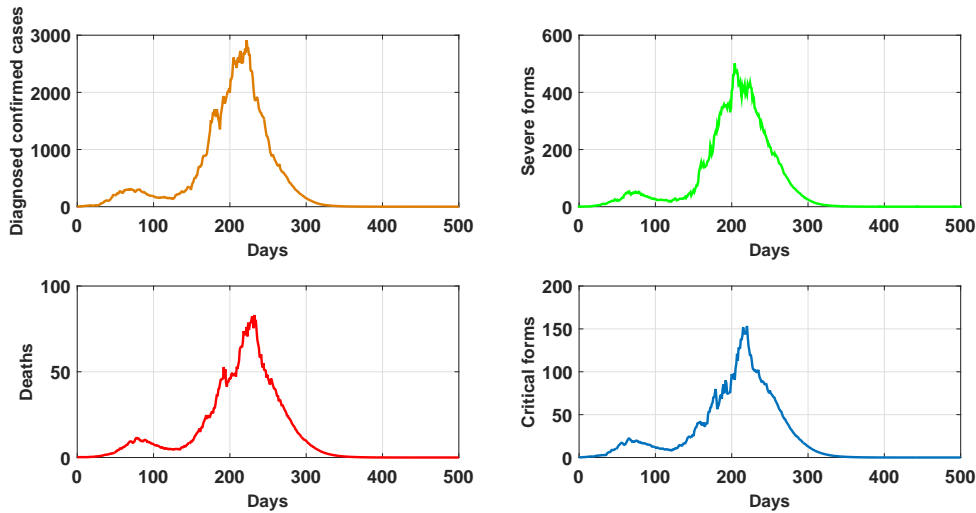}
\captionof{figure}{Evolution of the D-COVID-19 model \eqref{Sys1}
with deconfinement ($\rho=0.3)$ from May 20, 2020
and moderate effectiveness of the measures ($\sigma_2=0.15 $).}
\label{Deconfinement015}
\end{center}

\medskip

It should be mentioned that all our plots were obtained 
using the \textsf{Matlab} numerical computing environment 
by discretizing system \eqref{Sys1} 
by means of the higher order method of Milstein 
presented in \cite{Higham} and used in \cite{Mahrouf}.


\section{Conclusion}
\label{sec:5}

In this work, we have proposed a delayed stochastic mathematical model 
to describe the dynamical spreading of COVID-19 in Morocco 
by considering all measures designed by authorities, such as confinement 
and deconfinement policies. More precisely, our model takes into account 
four types of delays: the first one is related to the incubation period, 
the second is the time needed to move from the symptomatic infected individuals 
to the three forms of diagnosed cases, the third is the time needed 
to move from the class of infected individuals to the recovered or dead class, 
while the last one is the time needed to pass from the three types of classes
of individuals supported by the Moroccan health system, and under quarantine, 
to the recovered or dead compartments. Besides, to well describe reality, 
we have added a stochastic factor resulting from possible maladjustment
of the population individuals to the measures.

To show that our model is mathematically and biologically well-posed, 
we have proved the global existence of a unique positive solution
(see Theorem~\ref{thm:2.1}). Our result has shown a possible extinction 
of the disease when $\sigma_1^{2}$ is greater than a threshold parameter 
(see Theorem~\ref{thm:3.1}).

In addition, numerical simulations have been performed to forecast the evolution 
of COVID-19. More precisely, we have shown that the evolution of our D-COVID-19 model 
follows the tendency of daily reported confirmed cases in Morocco 
(see Figure~\ref{Historique}). Further, if Moroccan people would maintain, 
strictly, their confinement policy, we observe that the disease dies out 
around four months from March 2, 2020 (see Figure~\ref{Evolution}). 
On the other hand, in response to the decision of deconfinement represented 
by the liberation of the $30\%$ of population, which took place at May 20, 
we simulate three scenarios corresponding to different values of the intensity 
$\sigma_2$. When $\sigma_2= 0.01$ (Figure~\ref{Deconfinement001}), 
the eradication of the disease from the population comes early compared 
to the cases when $\sigma_2=0.10$ (Figure~\ref{Deconfinement01}) 
and  $\sigma_2=0.15$ (Figure~\ref{Deconfinement015}). Additionally, 
the number of diagnosed confirmed cases mainly changes because of the value
of this intensity and a small perturbation leads to relevant quantitative 
changes and significant variations on the time needed for extinction. 
Thus, we observe that the value of this perturbation has a high impact 
on the evolution of COVID-19, which means the Moroccan population has a 
big interest to respect the governmental measures announced May 20, 2020, 
in order to have a successful and good deconfinement strategy.

Here we have compared the predictions of the proposed model \eqref{Sys1} 
with real data until middle of May 2020. We leave the comparison  
of the real data in Morocco till the end of 2020 to a future work, 
where we also plan to incorporate the predictions of the evolution 
of our COVID-19 model with respect to preventive Moroccan measures 
by regions and cities.


\begin{acknowledgement}
The authors are grateful to an Associate Editor and two anonymous referees, 
who kindly reviewed an earlier version of the manuscript and provided 
several valuable suggestions and comments.
\end{acknowledgement}



\end{document}